\documentclass[conference,letterpaper]{IEEEtran} 
\usepackage{graphicx}
\usepackage[cmex10]{amsmath}
 \usepackage{amsthm} 
\usepackage{amsfonts}
\usepackage{amssymb} 
\usepackage{graphicx}%
\usepackage{color}%
\usepackage{algorithmic}
 \usepackage{algorithm}
\newtheorem{mytheorem}{Theorem}

\usepackage[english]{babel}
\usepackage{blindtext}
\usepackage{epsfig}
\title{Upper Bounds and Duality Relations of the Linear Deterministic Sum Capacity for Cellular Systems}

\author{
\IEEEauthorblockN{Rick Fritschek}
\IEEEauthorblockA{ Lehrstuhl f\"ur Informationstheorie 
    und \\ Theoretische
    Informationstechnik\\
    Technische Universit\"at Berlin, \\
    Einsteinufer 25,
    D--10587 Berlin, Germany\\
    Email: rick.fritschek@tu-berlin.de
}%
\and
\IEEEauthorblockN{Gerhard Wunder}
\IEEEauthorblockA{Fraunhofer
  Heinrich--Hertz--Institut\\
Wireless Communication and Networks \\
Einsteinufer 37, D--10587 Berlin, Germany\\
Email: gerhard.wunder@hhi.fraunhofer.de}

}

\begin{document}

\maketitle
\begin{abstract}
The MAC-BC duality of information theory and wireless communications is an intriguing concept for efficient algorithm design. However, no concept is known so far for the important cellular channel. To make progress on this front, we consider in this paper the linear deterministic cellular channel. In particular, we prove duality of a network with two interfering MACs in each cell and a network with two interfering BCs in each cell. The operational region is confined to the weak interference regime. First, achievable schemes as well as upper bounds will be provided. These bounds are the same for both channels. We will show, that for specific cases the upper bound corresponds to the achievable scheme and hence establishing a duality relationship between them.
\end{abstract}

\section{Introduction}

Information-theoretic problems for BC models are in general much harder to solve, than for the MAC models. This fact 
motivates a research which investigates possible relations between the two channel models. 
These relations are usually referred to as \textit{duality} or \textit{reciprocity}. Duality could mean that the capacity region, the achievable scheme or the upper bound of an BC model
is the same as in the dual-MAC set-up. 
%An capacity relationship contains connections between the achievable schemes as well as upper bounds. 
%Such a connection between dual channel models can be extremely useful. 
This could help to calculate hard BC problems, by transferring them to the rather easy dual-MAC model. And translating
the solution back to the BC model. A dual model is defined in the usual way, as a network with the same nodes and link gains as in the original channel model, but reversed directions of the transmissions. This means that the nodes interchange their original purpose of transmitter and receiver.
In recent years, a lot of progress has been made towards the understanding of duality.
Among the most notable results is \cite{Jindal2004}. Here it has been shown, that the capacity region for the Gaussian BC is the union over all
Gaussian MAC capacity regions of individual power constraints which add up to the BC power constraint. Likewise the capacity region of the Gaussian MAC is equal to the intersection of the dual Gaussian BC capacity regions. 
So the only difference is the distribution of power constraints. That is jointly in the BC case or distributed over all paths in the MAC case.
Another result for duality was established in \cite{Bresler2010}, in which the model of investigation is the Many-to-One and One-to-Many Gaussian interference channel.
It was found that under the approximation of the Linear Deterministic Model (LDM), first introduced in \cite{Avestimehr2007}, the capacity regions are identical. Therefore showing the duality of the two channel models. The LDM approximation is a powerful technique which can be used to investigate certain models under approximate conditions. These conditions are e.g. truncation at noise level, which simplifies calculation of the capacity region in contrast to the Gaussian case. However, these approximate results can be used to shine light on possible solutions of the Gaussian case. For example \cite{Bresler2008} showed, that the capacity region of the linear deterministic interference channel is within a constant 42 bit gap of the corresponding Gaussian interference channel. 

 \begin{figure}
   \centering
   \includegraphics[width=0.43 \textwidth]{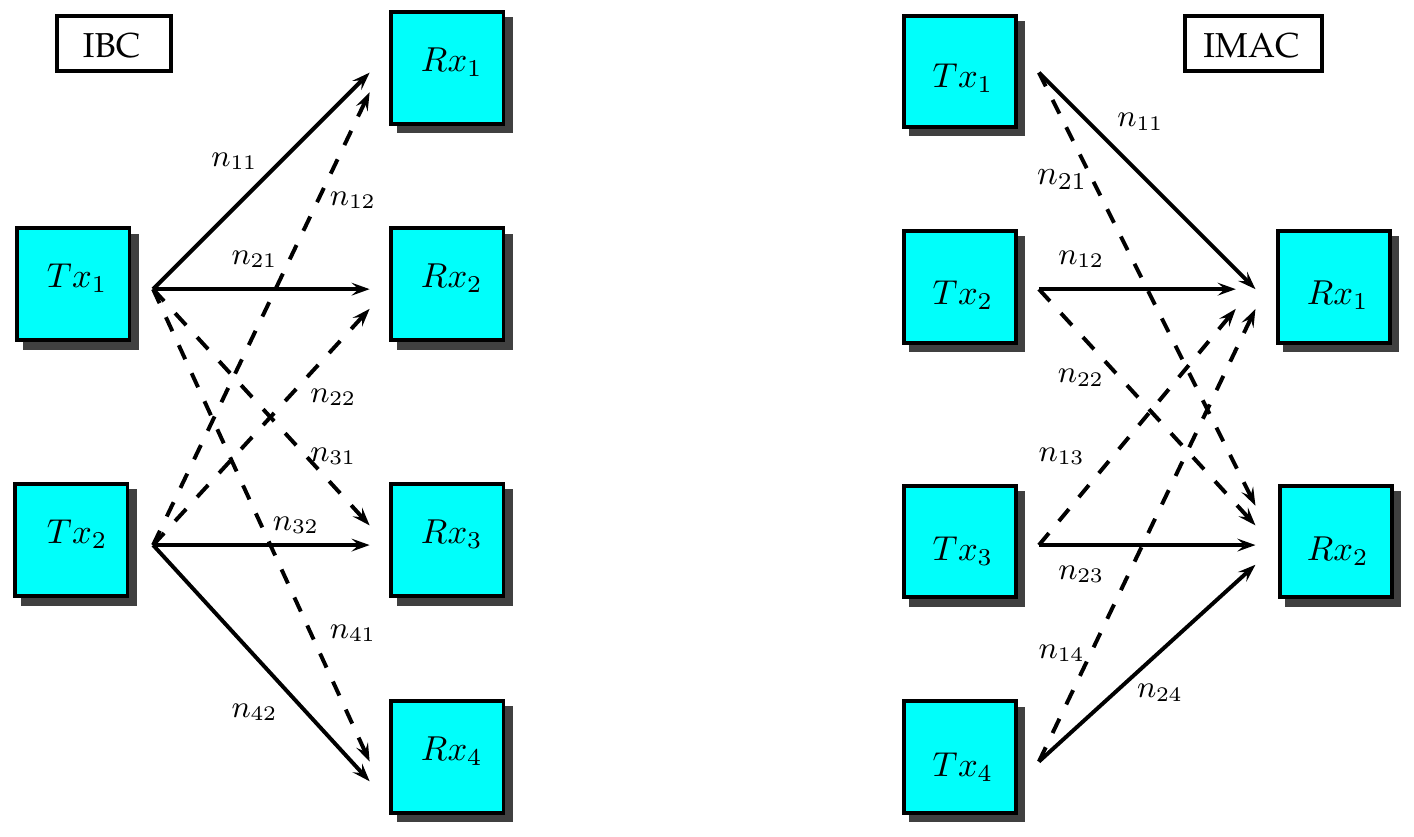}
   \caption{IBC and IMAC Model: Solid lines represent the direct signals and dotted lines the interference.}
   \label{system model}
 \end{figure}
 
{\bf Contributions}
In this paper we investigate two models of cellular networks, namely the interfering-MAC system, which will be henceforth referred to as the IMAC following the naming of \cite{Suh2008} and the dual model, the interfering-BC (IBC) system. The IMAC is a cellular system which consists of two cells with a MAC in each cell, interfering among themselves. The same for the IBC system, which are two broadcast channels separated in cells and interfering with each other. \cite{Suh2008} investigated the Gaussian IMAC under multipath and delay conditions and derived an DoF result, which states that the interference-free DoF can be reached as the number of transmitters increases in each cell. A duality result for the DoF between the IMAC and the IBC is also derived. But no capacity result was shown. We take another approach and investigate the approximative LDM IMAC and IBC model.
As to the best of our knowledge, no other capacity duality result for a multi-cell LDM system is known. We will therefore derive the sum capacity for the IBC and the IMAC for the very weak interference regime and show that under specific conditions a duality relationship exists. The investigated channel models include the BC-P2P (which consists of a BC interfering with a Point-to-Point Channel) and the previously in \cite{Buhler2011}  investigated MAC-P2P as special cases and therefore show a duality result for these models as well. We will also expand the proof techniques of the IMAC to the k-transmitter IMAC. Here we show that with increasing number of transmitters the capacity approaches the interference-free case.

{\bf Basic Notation}
For now on all vectors and matrices are written bold. Random variables are written upper case and values of these variables lower case. The elements of the vectors and matrices are elements of the finite binary field $\mathbb{F}_2$. 
 
To specify a particular range of elements in a bit level vector we use the notation $\mathbf{A}_{[i:j]}$ to indicate that $\mathbf{A}$ is restricted to the bit levels $i$ to $j$. If $i=1$, it will be omitted $\mathbf{A}_{[:j]}$, the same for $j\!=\!n$ $\mathbf{A}_{[i:]}$. Therefore $\mathbf{A}=\mathbf{A}_{[:]}$ which would correspond to no restriction at all.

\section{System Model} 
\label{sec:Setup}

\subsection{IMAC Model}

The system consists of 4 transmitters and 2 receivers. Transmitters $Tx_1$ and $Tx_2$ together with the receiver $Rx_1$ and $Tx_3$, $Tx_4$ with $Rx_2$ each form a MAC and both are interfering with each other (see fig. \ref{system model}).

A $(2^{nR_1}, 2^{nR_2}, 2^{nR_3}, 2^{nR_4},n)$ code will consist of four encoding and two decoding functions. The encoder $i$ assigns a codeword $x_i^n(m_i)$ to each message $m_i\in[1\!:\!2^{nR_i}]$ and the associated decoder $k$ assigns an estimate $(\hat{m}_k,\hat{m}_{k+1})\in [1\!:\!2^{nR_k}]\times [1\!:\!2^{nR_{k+1}}]$ for $k\in \{1,3\}$.
The probability of error will be defined as $P_e=P((\hat{m}_{k},\hat{m}_{k+1})\!\neq\! (m_k,m_{k+1}))$. We assume that the message pairs are uniformly distributed over $[1\!:\!2^{nR_k}]\times [1\!:\!2^{nR_{k+1}}]$ and independent of each other. A rate pair is said to be achievable if there exists a sequence of $(2^{nR_1},2^{nR_2},2^{nR_3},2^{nR_4},n)$ codes for which $\lim_{n\rightarrow\infty} P_e=0$.

As additional modification to simplify the system model the Linear Deterministic Model (LDM) is used.
 The LDM models the input symbols at $Tx_i$ as bit vectors $\mathbf{x}_i$. This is achieved by a binary expansion of the real input signal. The resulting bits constitute the new bit vector. The positions within the vector will be referred to as levels. To model the signal impairment induced by noise, the bit vectors will be truncated at noise level and only the n most significant bits are received at $Rx_i$. This is done by shifting the incoming bit vector for $q-n$ positions $\mathbf{Y}=\mathbf{S}^{q-n}\mathbf{X}$.
Where $\mathbf{S}$ is the shift matrix defined as 
\begin{equation}
\mathbf{S}=\begin{pmatrix}
0 & 0 &  \cdots & 0 & 0\\
1 & 0 &  \cdots & 0 & 0\\
0 & 1 &  \cdots & 0 & 0\\
\vdots & \vdots & \ddots & \vdots & \vdots \\
0 & 0 &  \cdots & 1 & 0\\
\end{pmatrix}.
\end{equation}
Superposition at the receivers is modelled via binary addition of the incoming bit vectors on the individual levels. Carry over is not used to limit the superposition on the specific level where it occurs. The channel gain is represented by $n_{ij}$-bit levels which corresponds to $\lceil\log \mbox{SNR}\rceil$ of the original channel. 
With this definitions the model can be written as
\begin{IEEEeqnarray*}{rCl}
\mathbf{Y}_1 & = & \mathbf{S}^{q-n_{11}}\mathbf{X}_1\oplus \mathbf{S}^{q-n_{12}}\mathbf{X}_2\oplus \mathbf{S}^{q-n_{13}}\mathbf{X}_3\oplus \mathbf{S}^{q-n_{14}}\mathbf{X}_4\\
\mathbf{Y}_2 & = & \mathbf{S}^{q-n_{21}}\mathbf{X}_1\oplus \mathbf{S}^{q-n_{22}}\mathbf{X}_2\oplus \mathbf{S}^{q-n_{23}}\mathbf{X}_3\oplus \mathbf{S}^{q-n_{24}}\mathbf{X}_4 \IEEEyesnumber\label{MAC-MAC}
\end{IEEEeqnarray*}

with $n_{ij}$ as in fig. \ref{system model} shown.
The direct signals are for simplicity written as $n_{11}=n_1$, $n_{12}=n_2$, and $n_{23}=n_3$, $n_{24}=n_4$. It is assumed that $n_1\geq n_2$,  $n_3\geq n_4$ and the difference between the two signals is denoted as $n_1-n_2=:\Delta_1$ and $n_3-n_4=:\Delta_2$. Furthermore it is assumed that $n_{21}=n_{22}=:n_M$ and $n_{13}=n_{14}=:n_D$, stating that the interference caused by $\mathbf{X}_{i,j}$ at the receivers is the same. Note that this restriction is justified in the case, when the distance between the two cells is much bigger than the cell dimensions itself.

\subsection{IBC Model}
The IBC system consists of 2 transmitters and 4 receivers. Transmitter $Tx_1$ together with the receivers $Rx_1$ and $Rx_2$, $Tx_2$ with $Rx_3$ and $Rx_4$ form a BC and both are interfering with each other (see fig. \ref{system model}). A $(2^{nR_{12}}, 2^{nR_{34}}, 2^{nR_1},2^{nR_2},2^{nR_3},2^{nR_4},n)$ code will consist of two encoding and four decoding functions. Encoders 1 and 2 assign a codeword $x_1^n(m_{12},m_1,m_2)$ and $x_2^n(m_{34},m_3,m_4)$ respectively, to each message triple
 $(m_{12},m_1,m_2)\in[1:2^{nR_{12}}]\times[1:2^{nR_1}]\times [1:2^{nR_2}]$ and $(m_{34},m_3,m_4)\in[1:2^{nR_{34}}]\times[1:2^{nR_3}]\times [1:2^{nR_4}]$.
The associated decoder $i$ assigns an estimate $(\hat{m_{0i}},\hat{m_i})\in [1:2^{nR_{0i}}]\times[1:2^{nR_i}]$ to each received sequence $y_i$.
The probability of error will be defined as \begin{equation*} 
 P_e=P((\hat{m}_{0i},\hat{m}_i)\neq (m_{k},m_i))
\end{equation*} with $k \in \{12,34\}$ and $i$ for the corresponding cell and decoder, respectively. We assume that the message triples are uniformly distributed over the corresponding message sets and independent of each other. A rate pair is said to be achievable if there exists a sequence of $(2^{nR_{12}}, 2^{nR_{34}},2^{nR_1},2^{nR_2},2^{nR_3},2^{nR_4},n)$ codes for which $\lim_{n\rightarrow\infty} P_e=0$.

Just as the IMAC case, in the IBC system model we will use the LDM model to simplify the problem. Therefore the channel model can be written as:

\begin{IEEEeqnarray*}{rCl}
\mathbf{Y}_1 & = & \mathbf{S}^{q-n_{11}}\mathbf{X}_1\oplus \mathbf{S}^{q-n_{12}}\mathbf{X}_2\\
\mathbf{Y}_2 & = & \mathbf{S}^{q-n_{21}}\mathbf{X}_1\oplus \mathbf{S}^{q-n_{22}}\mathbf{X}_2\IEEEyesnumber\label{BC-BC}\\
\mathbf{Y}_3 & = & \mathbf{S}^{q-n_{32}}\mathbf{X}_2\oplus \mathbf{S}^{q-n_{31}}\mathbf{X}_1\\
\mathbf{Y}_4 & = & \mathbf{S}^{q-n_{42}}\mathbf{X}_2\oplus \mathbf{S}^{q-n_{41}}\mathbf{X}_1
\end{IEEEeqnarray*}

with $n_{ij}$ as in fig. \ref{system model} shown.
The direct signals are for simplicity written as $n_{11}=n_1$, $n_{21}=n_2$, and $n_{32}=n_3$, $n_{42}=n_4$. It is assumed that $n_1\geq n_2$,  $n_3\geq n_4$ and the difference between the two signals is denoted as $n_1-n_2=:\Delta_1$ and $n_3-n_4=:\Delta_2$. Furthermore it is assumed that $n_{31}=n_{41}=:n_M$ and $n_{12}=n_{22}=:n_D$, stating that the interference caused by $x_{i,j}$ at the receivers is the same. As in the IMAC case this restriction is justified when the distance between the cells is much bigger than the cell dimensions itself.

\section{Coding Schemes for the Very Weak Interference Case}
For better definitions of the particular ranges some definitions will be introduced. 
For the very weak interference case it is assumed that the sum of both interference parts of the signals are below the direct signal level. This is stated in the condition that $n_M+n_D \leq \min (n_1,n_3)$. For a symmetric model, this condition becomes: $0 \leq \alpha \leq \frac{1}{2}$, with $\alpha := \frac{n_i}{n_1}$.\\

{\bf IMAC System}

The achievability scheme for the IMAC is basically an extended version of the scheme already used for the MAC-P2P in \cite{Buhler2011}. 
Like in the MAC-P2P we split the system (\ref{MAC-MAC}) into two sub systems, $\mathcal{R}_{ach}^{(1)}$ and $\mathcal{R}_{ach}^{(2)}$. The sum of the achievable rates of these two sub systems will constitute the overall sum-capacity. The sub systems are equally structured and are given by the equations

\begin{IEEEeqnarray*}{rCl}
\mathbf{Y}_1^{(1)} & = & \mathbf{S}^{q^{(1)}-(n_1-n_D)}\mathbf{X}_1^{(1)}\oplus \mathbf{S}^{q^{(1)}-(n_2-n_D)}\mathbf{X}_2^{(1)}\\
\mathbf{Y}_2^{(1)} & = & \mathbf{S}^{q^{(1)}-n_M}\mathbf{X}_1^{(1)}\oplus \mathbf{S}^{q^{(1)}-n_M}\mathbf{X}_2^{(1)}\oplus \mathbf{S}^{q^{(1)}-n_M}\mathbf{X}_3^{(1)}\\
&&\oplus\: \mathbf{S}^{q^{(1)}-n_M}\mathbf{X}_4^{(1)} \IEEEyesnumber
\end{IEEEeqnarray*}

for $\mathcal{R}_{ach}^{(1)}$ and 

\begin{IEEEeqnarray*}{rCl}
\mathbf{Y}_1^{(2)} & = & \mathbf{S}^{q^{(2)}-(n_3-n_M)}\mathbf{X}_3^{(2)}\oplus \mathbf{S}^{q^{(2)}-(n_4-n_M)}\mathbf{X}_4^{(2)}\\
\mathbf{Y}_2^{(2)} & = & \mathbf{S}^{q^{(2)}-n_D}\mathbf{X}_1^{(2)}\oplus \mathbf{S}^{q^{(2)}-n_D}\mathbf{X}_2^{(2)}\oplus \mathbf{S}^{q^{(2)}-n_D}\mathbf{X}_3^{(2)}\\
&&\oplus\: \mathbf{S}^{q^{(2)}-n_D}\mathbf{X}_4^{(2)} \IEEEyesnumber
\end{IEEEeqnarray*}

for $\mathcal{R}_{ach}^{(2)}$. 

The achievable sum rates for the systems are defined as 

\begin{equation}
R_\Sigma^{(1)} \leq n_M + \zeta^{(1)} +\phi(n_M,\Delta_1)
\end{equation}
\begin{equation}
R_\Sigma^{(2)} \leq n_D + \zeta^{(2)}+\phi (n_D,\Delta_2).
\end{equation}
Where $\zeta^{(1)}:=n_2-n_M-n_D$, $\zeta^{(2)}:=n_4-n_M-n_D$ and the function $\phi$ for $p,q \in \mathbb{N}_0$, following the notation of \cite{Buhler2011}, defined as
\begin{equation}
\phi(p,q):=
\begin{cases}
q+\frac{l(p,q)q}{2} & \text{if } l(p,q)\text{ is even,}
\\
p-\frac{(l(p,q)-1)q}{2} & \text{if } l(p,q)\text{ is odd}.
\end{cases}
\label{phi}
\end{equation}
where $l(p,q):=\lfloor\frac{p}{q}\rfloor\ \mbox{for}\ q>0\ \mbox{and}\ l(p,0)=0$.
Considering the sub systems where $\mathbf{X}_3^{(1)}$ or $\mathbf{X}_4^{(1)}$ in $\mathcal{R}_{ach}^{(1)}$ and $\mathbf{X}_1^{(2)}$ or $\mathbf{X}_2^{(2)}$ in $\mathcal{R}_{ach}^{(2)}$ are removed one can see that the sum rates are achievable with interference alignment and optimal bit level assignment which is proven in   
 \cite{Buhler2011}. Therefore it is clear that the given sum rates are achievable for the subsystems as well. Finally, the sum rate for the overall system can be obtained by adding the sub systems: $R_\Sigma^{(1)}+R_\Sigma^{(2)} = R_\Sigma$
 \begin{IEEEeqnarray*}{rCl}
R_\Sigma &\leq & n_M + \zeta^{(1)} + n_D + \zeta^{(2)}+\phi(n_M,\Delta_1) +\phi (n_D,\Delta_2)\\
& = & n_2+n_4-n_M-n_D + \phi(n_M,\Delta_1) +\phi (n_D,\Delta_2)\IEEEyesnumber\label{IMAC-ach}
\end{IEEEeqnarray*}
\begin{figure}
\centering
\includegraphics[scale=0.65]{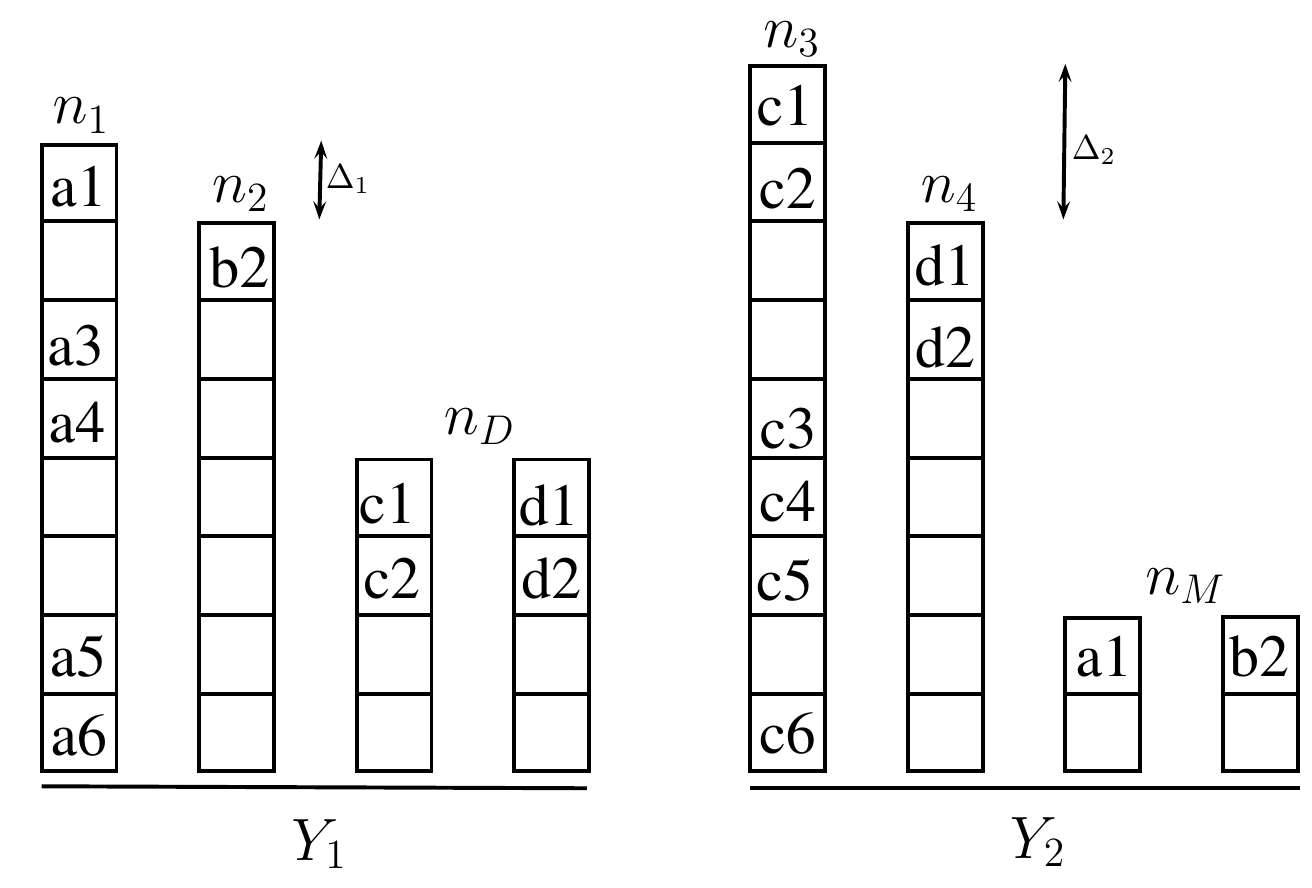}
\caption{An example for a scheme which achieves the upper bound is presented in the figure. The MAC-cell has $n_1=8$ bit levels and $n_2=7$ bit levels and generates interference, at the other side, of $n_M=2$ bit levels. Whereas the other MAC cell has $n_3=9 $ and $n_4= 7$ bit levels and generates $n_D=4$ bit levels interference. The scheme above yields a sum rate of 14 bit levels and therefore reaches the upper bound which can be calculated with $R_{\Sigma}\leq n_1+n_3-\frac{n_M}{2}-\frac{n_D}{2}$.}
\label{MAC_MAC_SCHEMA}
\end{figure}

{\bf IBC System}

Since the sum rate (\ref{IMAC-ach}) is achievable in the IMAC, it follows that the same sum rate can be achieved in the IBC. This is because the achievable scheme in the IMAC is linear and \cite{Raja2009} has proven that every linear coding scheme also proofs achievability for the dual case. A strategy based on the IMAC case would be to merge the dual MAC signals into one BC signal. So for a specific code vector we have  $\mathbf{x}_1^{BC}=\mathbf{x}_1^{MAC}\oplus\mathbf{x}_2^{MAC}$ and turn the resulting BC-bit vector upside down. The dual-case example for the one in the IMAC case is shown in figure \ref{BC-BC_Schema}.
\begin{figure}
\centering
\includegraphics[scale=0.61]{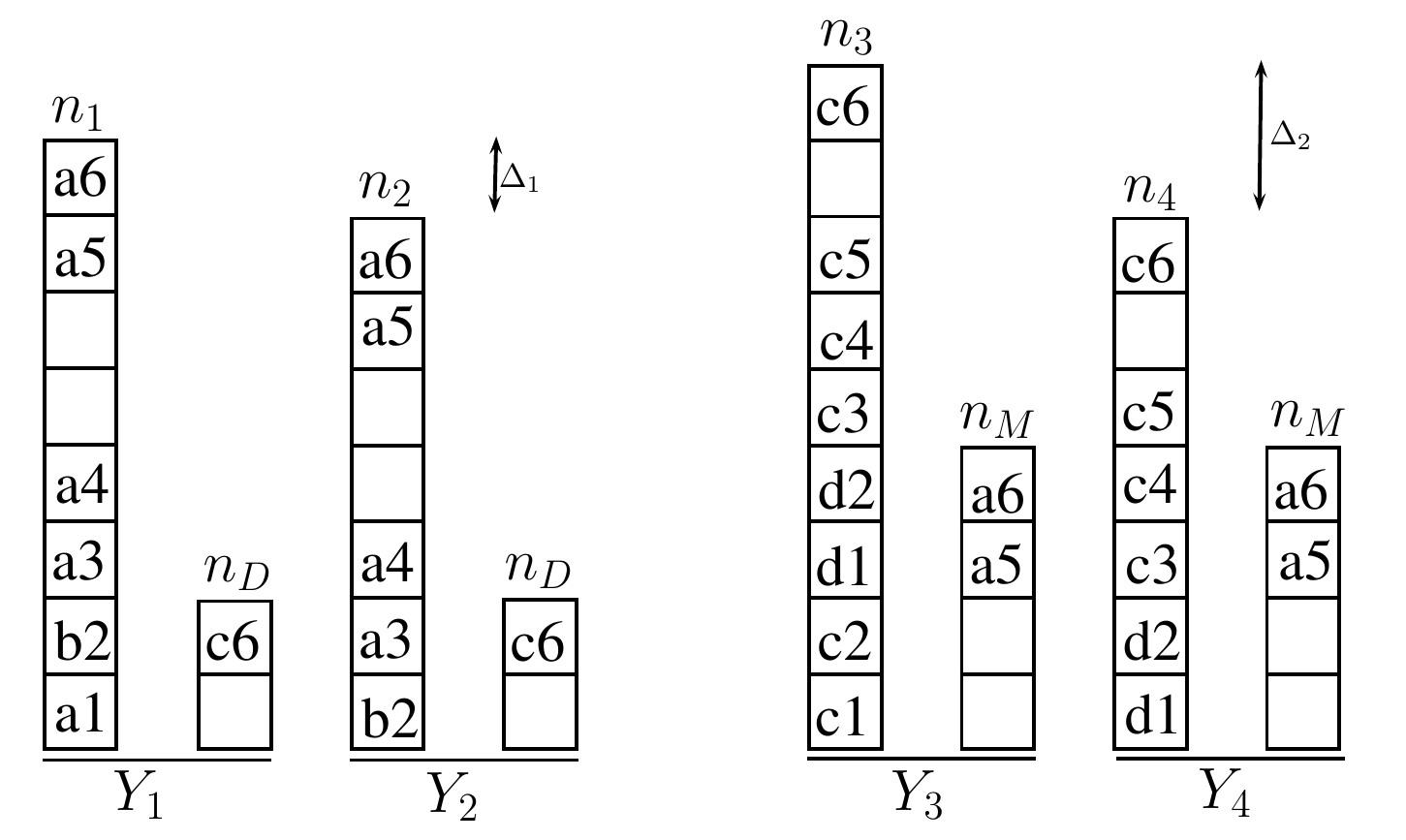}
\caption{An example for a coding scheme of the IBC system model is shown. The example is chosen such that it depicts the exact dual case to the IMAC model as in the previous example. Here one can see the basic strategy to get a coding scheme for the IBC case out of the IMAC case. The MAC components got merged and the coding vectors inverted as described in the section on the achievability scheme. Like in the IMAC model, the sum rate is 14 bit levels. }
\label{BC-BC_Schema}
\end{figure}

\section{Upper Bounds on Sum Rate}

To show the upper bounds, we need to divide the very weak interference case in two sub cases. 
The first one can be defined by $n_M+n_D \leq n_2$ and $n_M+n_D \leq n_4$ which basically prevents that the common part in the direct signals overlap with the interference of the other cell.
The second case is $n_2 \leq n_M+n_D \leq n_1$ and $n_4 \leq n_M+n_D \leq n_3$ which only prevents overlapping between the common part of the strongest direct signal and the interference. This condition is the same as limiting the overlapping part of the smaller direct signals with the interference by $\Delta$. For clarity of the exposition we will consider only the cases, when both cells fall in the same interference regime. However, the same proof techniques can be used to show the mixed cases as well.\\
Some additional notations for specific bit vector parts will be used. In the context of the IMAC, a bit vector $\mathbf{A}$ will be split up in a part $\mathbf{A}^{\uparrow}$ which is not affected by interference and a part $\mathbf{A}^{\downarrow}$. For example the bit vector $\mathbf{Y}_1$ can be split in $\mathbf{Y}_1^{\uparrow}=\mathbf{X}_{1,[:n_M+\Delta_1]} \oplus \mathbf{X}_{2,[:n_M]} $ and $\mathbf{Y}_1^{\downarrow}$ is the other part of $\mathbf{Y}_1$. 

For the IBC proofs, the split ranges are flipped, which means for example for the bit vector $\mathbf{Y}_1^{\uparrow}=\mathbf{X}_{1,[:n_1-(n_D+\Delta_1)]}$. For comparison, in the case of $n_M+n_D \leq n_2$ and $n_M+n_D \leq n_4$, the IMAC bit vector with the up-arrow $\mathbf{Y}_1^{\uparrow}$ has $n_M+\Delta_1$ bit-levels and the down-arrow vector $\mathbf{Y}_1^{\downarrow}$ has $n_2-n_M$. Whereas in the IBC system the up-arrow vector $\mathbf{Y}_1^{\uparrow}$ has $n_2-n_D$ bit-levels and the down-arrow vector $\mathbf{Y}_1^{\downarrow}$ has $n_D+\Delta_1$ bit-levels.
The common part of a bit vector $\mathbf{A}$ is denoted as $\hat{\mathbf{A}}$, for example: $\hat{\mathbf{X}}_1=\mathbf{X}_{1,[:n_M]}$. \\

{\bf IMAC: Upper Bound}

\begin{mytheorem}[]
The sum rate for the IMAC system model can be upper bounded by 
\begin{equation}
R_\Sigma \leq n_1+n_3-\frac{n_M}{2}-\frac{n_D}{2}.
\end{equation}
\end{mytheorem}

\begin{proof}
Considering Fano's inequality and the Data Processing inequality one can establish the following bounds:
\begin{IEEEeqnarray*}{rCl}
\IEEEeqnarraymulticol{3}{l}{
n(R_1+R_2+R_3+R_4)}\\
& \leq & I(\mathbf{X}_1^n,\mathbf{X}_2^n;\mathbf{Y}_1^n)+I(\mathbf{X}_3^n,\mathbf{X}_4^n; \mathbf{Y}_2^n) +n(\epsilon_{n,12}+\epsilon_{n,34})\\
& = & H(\mathbf{Y}_1^n)-H(\mathbf{Y}_1^n| \mathbf{X}_1^n,\mathbf{X}_2^n)+H(\mathbf{Y}_2^n)\\
&&-\: H(\mathbf{Y}_2^n|\mathbf{X}_3^n,\mathbf{X}_4^n)+n(\epsilon_{n,12}+\epsilon_{n,34}).
\end{IEEEeqnarray*}
For the case of $n_M+n_D \leq n_2$ and $n_M+n_D \leq n_4$ we begin by 
\begin{IEEEeqnarray*}{rCl}
\IEEEeqnarraymulticol{3}{l}{
2 n(R_\Sigma - \epsilon_{n,\Sigma})}\\
 &\leq& 2 H(\mathbf{Y}_1^n)-2 H (\mathbf{Y}_1^n| \mathbf{X}_1^n,\mathbf{X}_2^n) +2 H(\mathbf{Y}_2^n)\\
 &&-\: 2H(\mathbf{Y}_2^n|\mathbf{X}_3^n,\mathbf{X}_4^n)\\
& = & 2 H(\mathbf{Y}_1^{n,\downarrow})+2 H(\mathbf{Y}_1^{n,\uparrow})-2H(\mathbf{\hat{X}}_1^n \oplus \mathbf{\hat{X}}_2^n)\\
&&+\: 2 H(\mathbf{Y}_2^{n,\downarrow})+2 H(\mathbf{Y}_2^{n,\uparrow})-2H(\mathbf{\hat{X}}_3^n \oplus \mathbf{\hat{X}}_4^n)\\ 
& \overset{(a)}{\leq} & 2n(n_2-n_M) + H(\mathbf{Y}_1^{n,\uparrow})+ n\Delta_1\\
&& +\: 2n(n_3-n_D) + H(\mathbf{Y}_2^{n,\uparrow})+n\Delta_2\\ 
& \leq & 2n(n_2-n_M)+2n \Delta_1 + nn_M\\
&& +\: 2n(n_4-n_D) + 2n\Delta_2 + nn_D
\end{IEEEeqnarray*}
where we used that $H(\mathbf{Y}_{1}^{n,\uparrow})\leq n(n_{m}+\Delta_{1})$ and $H(\mathbf{Y}_{3}^{n,\uparrow}) \leq n(n_{D}+\Delta_{2})$. Also (a) can be shown by considering
\begin{IEEEeqnarray*}{rCl}
\IEEEeqnarraymulticol{3}{l}{
H(\mathbf{Y}_{1}^{n,\uparrow})-2H(\mathbf{\hat{X}}_{1}^n \oplus \mathbf{\hat{X}}_{2}^n)}\\
& \leq & H(\mathbf{Y}_{1}^{n,\uparrow}) -H(\mathbf{\hat{X}}_{1}^n \oplus \mathbf{\hat{X}}_{2}^n|\mathbf{\hat{X}}_{1}^n)-H(\mathbf{\hat{X}}_{1}^n \oplus \mathbf{\hat{X}}_{2}^n|\mathbf{\hat{X}}_{2}^n) \\
& \leq &  n(\Delta_{1})
\end{IEEEeqnarray*}
 which can be also shown for the other MAC cell using the independence between the direct bit vectors.
Dividing both sides by two and taking $n\rightarrow \infty$ yields the desired upper bound.
The second case ($n_2 \leq n_M+n_D \leq n_1$ and $n_4 \leq n_M+n_D \leq n_3$) can be shown similarly.\\
\end{proof}

{\bf IBC: Upper Bound}

\begin{mytheorem}[]
The sum rate for the IBC system model can be upper bounded by 
\begin{equation}
R_\Sigma \leq n_1+n_3-\frac{n_M}{2}-\frac{n_D}{2}.
\end{equation}
\end{mytheorem}

\begin{proof}
One can show that (see also \cite{Gamal1979})
\begin{IEEEeqnarray*}{rCl}
\IEEEeqnarraymulticol{3}{l}{
n(R_{12}+R_1+R_2)}\\
& = & H(M_{12},M_1) + H(M_{12},M_2) - H(M_{12})\\
& = & I(M_{12},M_1;\mathbf{Y}_1^n) + I(M_{12},M_2;\mathbf{Y}_2^n) - I(M_{12};\mathbf{Y}_1^n)\\
&& +\: H(M_{12},M_1|\mathbf{Y}_1^n) + H(M_{12},M_2|\mathbf{Y}_2^n)\IEEEyesnumber\\
&& -\: H(M_{12}|\mathbf{Y}_1^n)\\
&\leq & I(M_{12},M_1;\mathbf{Y}_1^n) + I(M_{12},M_2;\mathbf{Y}_2^n) - I(M_{12};\mathbf{Y}_1^n)\\
&& +\: n(\epsilon_{1n} + \epsilon_{2n})\\
& = & I(M_1;\mathbf{Y}_1^n|M_{12}) + I(M_{12},M_2;\mathbf{Y}_2^n) + n(\epsilon_{1n}+\epsilon_{2n})
\end{IEEEeqnarray*}
using Fano's inequality, independence of $M_{12},M_1,M_2$, chain rule and the definition of mutual information. The same can be shown for $\mathbf{Y}_3^n$ and $\mathbf{Y}_4^n$.
Utilizing this relationship and combining the two results one obtains for the case $n_M+n_D \leq n_2$ and $n_M+n_D \leq n_4$:
\begin{IEEEeqnarray*}{rCl}
\IEEEeqnarraymulticol{3}{l}{
n(R_{12}+R_{34}+R_1+R_2+R_3+R_4-\epsilon_{n,\Sigma})}\\
 & \leq & I(M_1; \mathbf{Y}_1^n|M_{12})+I(M_{12},M_2; \mathbf{Y}_2^n)\\
 &&+\: I(M_3; \mathbf{Y}_3^n |M_{34})+I(M_{34},M_4;\mathbf{Y}_4^n)\\
&\overset{(a)}{\leq}& I(M_1; \mathbf{Y}_1^n|M_{12},M_2)+I(M_{12},M_2; \mathbf{Y}_2^n)\\
&&+\: I(M_3; \mathbf{Y}_3^n |M_0,M_4)+I(M_{34},M_4;\mathbf{Y}_4^n)\\
%& = & H(\mathbf{Y}_1^n |M_{12}, M_2)-H(\mathbf{Y}_1^n |M_{12},M_1,M_2)+ H(\mathbf{Y}_2^n)\\
%&&-\: H(\mathbf{Y}_2^n|M_{12},M_2)+ H(\mathbf{Y}_3^n |M_{34}, M_4)\\
%&&-\: H(\mathbf{Y}_3^n |M_{34},M_3,M_4)+H(\mathbf{Y}_4^n)-H(\mathbf{Y}_4^n | M_{34},M_4)\\
& = & H(\mathbf{Y}_1^n |M_{12}, M_2)-H(\mathbf{Y}_1^n |\mathbf{X}_1^n)+ H(\mathbf{Y}_2^n)\\
&&- \: H(\mathbf{Y}_2^n|M_{12},M_2)+ H(\mathbf{Y}_3^n |M_{34}, M_4)\\
&&- \: H(\mathbf{Y}_3^n |\mathbf{X}_2^n)+H(\mathbf{Y}_4^n)-H(\mathbf{Y}_4^n | M_{34},M_4)
\end{IEEEeqnarray*}
where we used the independence of $M_{12},M_1,M_2$ and $M_{34},M_3,M_4$ again in (a). Adding two of the inequalities yields
\begin{IEEEeqnarray*}{rCl}
\IEEEeqnarraymulticol{3}{l}{
2n(R_\Sigma-\epsilon_{n,\Sigma}) }\\
& \leq &  2H(\mathbf{Y}_1^n | M_{12},M_2)-2H(\mathbf{Y}_1^n | \mathbf{X}_1^n)-2H(\mathbf{Y}_2^n| M_{12},M_2)\\
&&+\: 2H(\mathbf{Y}_2^n)+2H(\mathbf{Y}_4^n) +2H(\mathbf{Y}_3^n |M_{34}, M_4)\\
&&-\: 2 H(\mathbf{Y}_3^n |\mathbf{X}_2^n)-2 H(\mathbf{Y}_4^n |M_{34}, M_4)\\
& \overset{(a)}{\leq} & 2n(n_2-n_M)+2n(n_4-n_D) \\
&&+ \: H(\mathbf{Y}_{1}^{n,\downarrow}|M_{12},M_{2},\mathbf{X}_{1}^{n,\uparrow})\\
&&+\: n\Delta_1 +  H(\mathbf{Y}_{3}^{n,\downarrow}|M_{34},M_{4},\mathbf{X}_{2}^{n,\uparrow})+ n\Delta_2\\
& \leq & 2n(n_2-n_M)+2n(n_4-n_D)\\
&&+\: 2n\Delta_1 + nn_M+nn_D+ 2n\Delta_2\IEEEyesnumber\label{BC-Bound}
\end{IEEEeqnarray*} with (a) following from
\begin{IEEEeqnarray*}{rCl}
\IEEEeqnarraymulticol{3}{l}{
2H(\mathbf{Y}_{1}^{n}|M_{12},M_{2})-2H(\mathbf{Y}_2^n| M_{12},M_2)}\\
& \leq &  2 H(\mathbf{Y}_{1}^{n,\downarrow}|M_{12},M_{2},\mathbf{X}_{1}^{n,\uparrow})\\
&&-\: H(\mathbf{X}_{1}^{n,\downarrow} \oplus \mathbf{\hat{X}}_{2}^n| M_{12},M_{2},\mathbf{X}_{1}^{n,\downarrow},\mathbf{X}_{1}^{n,\uparrow})\\\IEEEyesnumber\label{BC-remove}
&&-\:H(\mathbf{X}_{1}^{n,\downarrow}\oplus \mathbf{\hat{X}}_{2}^n| M_{12},M_{2},\mathbf{\hat{X}}_{2}^n,\mathbf{X}_{1}^{n,\uparrow})\\
& \leq & H(\mathbf{Y}_{1}^{n,\downarrow}|M_{12},M_{2},\mathbf{X}_{1}^{n,\uparrow})+n\Delta_{1}
\end{IEEEeqnarray*}
 since $\mathbf{X}_{1}^{n,\downarrow}$ and $\mathbf{\hat{X}}_{2}^n$ are independent. The same can be done with the $\mathbf{Y}_3^n$, $\mathbf{Y}_3^n$ terms.
 Dividing both sides of (\ref{BC-Bound}) by two and taking $n \rightarrow \infty$ results in desired upper bound. 
 
A similar strategy can show the upper bound for the second case ($n_2 \leq n_M+n_D \leq n_1$ and $n_4 \leq n_M+n_D \leq n_3$).\\
\end{proof}

{\bf K-Transmitter IMAC with Very Weak Interference}\\
An extension of the proof for the IMAC upper bound, is the case where each MAC cell has k-transmitters.

\begin{mytheorem}[]
The upper bound for the k-transmitter IMAC very weak interference case is
\begin{equation}
R_{\Sigma}\leq n_1-n_D+n_3-n_m+\frac{(k-1)n_D}{k}+\frac{(k-1)n_M}{k}.
\end{equation}
\end{mytheorem}

\begin{proof}
To show the upper bound for the case $n_M+n_D \leq n_2$ and $n_M+n_D \leq n_4$ one can take Fano's inequality as the starting point, like in the 2-sender case and obtain:

\begin{IEEEeqnarray*}{rCl}
\IEEEeqnarraymulticol{3}{l}{
n(\sum\limits_{i=1}^{2k} R_i -\epsilon_n) }\\
& \leq & I(\mathbf{X}_1^n,\mathbf{X}_2^n,\cdots ,\mathbf{X}_k^n; \mathbf{Y}_1^n) \\
&&+\: I(\mathbf{X}_{k+1}^n,\mathbf{X}_{k+2}^n,\cdots, \mathbf{X}_{2k}^n; \mathbf{Y}_2^n)\\
& = & H(\mathbf{Y}_1^n)-H(\mathbf{Y}_1^n| \mathbf{X}_1^n, \mathbf{X}_2^n, \cdots, \mathbf{X}_k^n)+H(\mathbf{Y}_2^n)\\
&&-\: H(\mathbf{Y}_2^n|\mathbf{X}_{k+1}^n, \mathbf{X}_{k+2}^n, \cdots, \mathbf{X}_{2k}^n)
\end{IEEEeqnarray*}

adding k of them will lead to the following bound:
\begin{IEEEeqnarray*}{rCl}
\IEEEeqnarraymulticol{3}{l}{
kn(\sum\limits_{i=1}^{2k} R_i -\epsilon_n)}\\
&\leq & kH(\mathbf{Y}_1^n)-kH(\mathbf{Y}_1^n| \mathbf{X}_1^n,\mathbf{X}_2^n,\cdots ,\mathbf{X}_k^n)+kH(\mathbf{Y}_2^n)\\
&&- \: kH(\mathbf{Y}_2^n|\mathbf{X}_{k+1}^n,\mathbf{X}_{k+2}^n,\cdots, \mathbf{X}_{2k}^n)\\
& = & k H(\mathbf{Y}_1^{n,\downarrow})+k H(\mathbf{Y}_1^{n,\uparrow})-kH(\bigoplus\limits_{i=1}^{k}\mathbf{\hat{X}}_i^n )+ k H(\mathbf{Y}_2^{n,\downarrow})\\
&&+ \: k H(\mathbf{Y}_2^{n,\uparrow})-kH(\bigoplus\limits_{i=k+1}^{2k}\mathbf{\hat{X}}_i^n)\\ 
& \leq & kn(n_2-n_M) + (k-1)H(\mathbf{Y}_1^{n,\uparrow})+ n\Delta_1 + kn(n_3-n_D)\\
&& +\: (k-1)H(\mathbf{Y}_2^{n,\uparrow})+n\Delta_2\\ 
& \leq & kn(n_2-n_M)+kn\Delta_1 + (k-1)nn_M + kn(n_3-n_D)\\
&& +\: kn\Delta_2 + (k-1)nn_D
\end{IEEEeqnarray*}
where we used that $H(\mathbf{Y}_{1}^{n,\uparrow})-kH(\bigoplus\limits_{i=1}^{k}\mathbf{\hat{X}}_{i}^n)\leq n(\Delta_{1})$ and the same for the $\mathbf{Y}_{2}^{n,\uparrow}$ terms. 
Dividing by nk and taking $n \rightarrow \infty$ yields the upper bound. 

The upper bound for the case ($n_2 \leq n_M+n_D \leq n_1$ and $n_4 \leq n_M+n_D \leq n_3$) follows on the same lines as in the 2-sender IMAC case, applied to k-senders as above. Since the proof can be split for the two sub systems, the proof for the mixed case can be shown in a similar manner.

\end{proof}

This bound can be achieved under the condition that all k-transmitters have different signal strength (i.e. bit-level amount) and every transmitter utilizes the highest $\Delta_i$ part of his signal. In this case the amount of interference at the opposite cell equals the largest $\Delta_i$. The minimum of interference can be reached if every one of the k-transmitters utilizes $\frac{n_i}{k}$ of the top most part, where $n_i$ stands for the interference strength at the opposite cell. Therefore, the $\Delta_i$-shifts must also be equally distributed to reach the maximum rate.
As the number of transmitters, k, is growing, more interfering signals align at a smaller range of bit levels. Therefore the sum capacity is approaching the interference-free sum capacity. The optimum case and maximum reachable rate in the linear deterministic case, would be to have the same amount of transmitters as bit-levels of the interference strength, which are equally distributed (1 bit minimum shift). In this case, the interference at an arbitrary signal strength in one cell would be 1 bit.

\section{Comparison of Cellular Systems with the Interference Channel}\label{Dis}

As already partially mentioned above, the interference channel was investigated in numerous papers for different interference regimes and scenarios. The question that now arises is if cellular channels are much different and why practically no results exists for cellular channels of the type mentioned in this paper. One can answer this question by looking into the differences of the achievable schemes and the proofs for converse arguments using the example of the MAC-P2P and IMAC system. For example in the very weak interference regime, the achievable scheme for the interference channel relies on relatively simple Han-Kobayashi like schemes. These schemes either make the entire signal private information or transmit only on the levels, which are not affected by noise \cite{Bresler2008}. The latter strategy imitates the technique of treating interference as noise. Examples for these coding schemes are relatively easy constructed and some can be found for instance in \cite{Bresler2008}.

In contrast, the MAC-P2P, IMAC and IBC cellular systems need a rather complicated interference alignment strategy, where the interference of two signals align on bit levels \cite{Buhler2012}. Here the shift between the direct signals in a cell needs to be exploited for the alignment. This shift is the main ingredient which enables the IMAC and IBC system to achieve higher rates than the interference channel. As pointed out above for the k-sender case, many senders can yield even higher gains and potentially approach interference-free sum capacity. An example for a scheme which exploits the signal-shift can be found in figure \ref{MAC_MAC_SCHEMA}. We therefore need to enable 
the {\it multi-user gain} \cite{Suh2008}, which is the main difference to the interference channel. 
Other advanced cellular systems probably need more complicated strategies to achieve the bounds and utilize this gain, where exploitation of shift-properties need to be combined with copy-bits. The need for strategies where codes utilize bit-level alignment of signals in addition to copy-bits can be seen with a relatively simple example for the symmetrical IMAC system where the direct signals equal the interference $n_1=n_i$ with $\alpha=1$. In the interference channel, this point constitutes the second minimum point of the w-curve, where the rate falls back to the single link rate. But for the IMAC system, a higher rate can be reached. Figure \ref{MAC-MAC_Schema_3} gives an example for the case of $n_1=4$.

\begin{figure}
\centering
\includegraphics[scale=0.9]{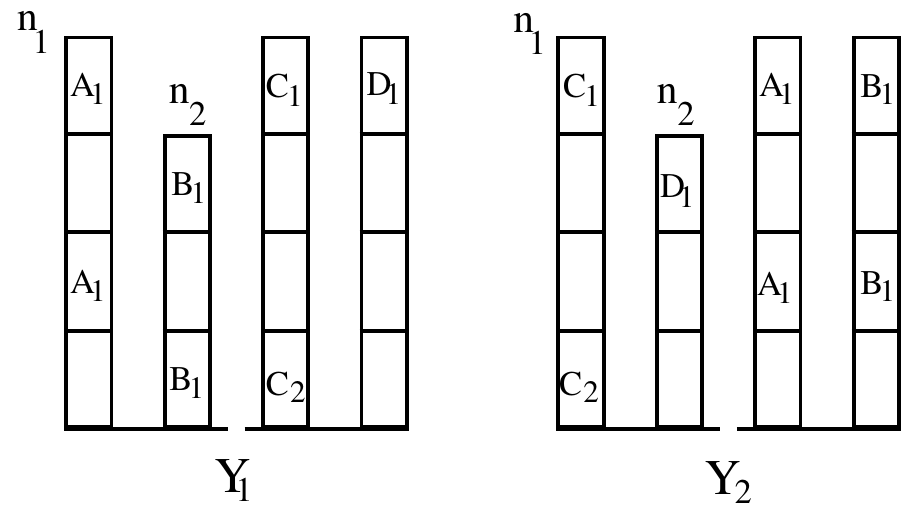}
\caption{A coding scheme for the case of $n_1=n_i=4$ with $\alpha=1$. $A_1$ and $B_1$ are used as copy-bits and the shift-property is used at both sides. The achievable rate of the example is 5 bit level whereas the upper bound for the interference channel lies at 4 bit level.}
\label{MAC-MAC_Schema_3}
\end{figure}

 This features make investigations in cellular systems challenging but also provide a chance for new coding schemes with higher gains for future networks.

\section{Conclusions}
\begin{figure}
\centering
\includegraphics[scale=0.42]{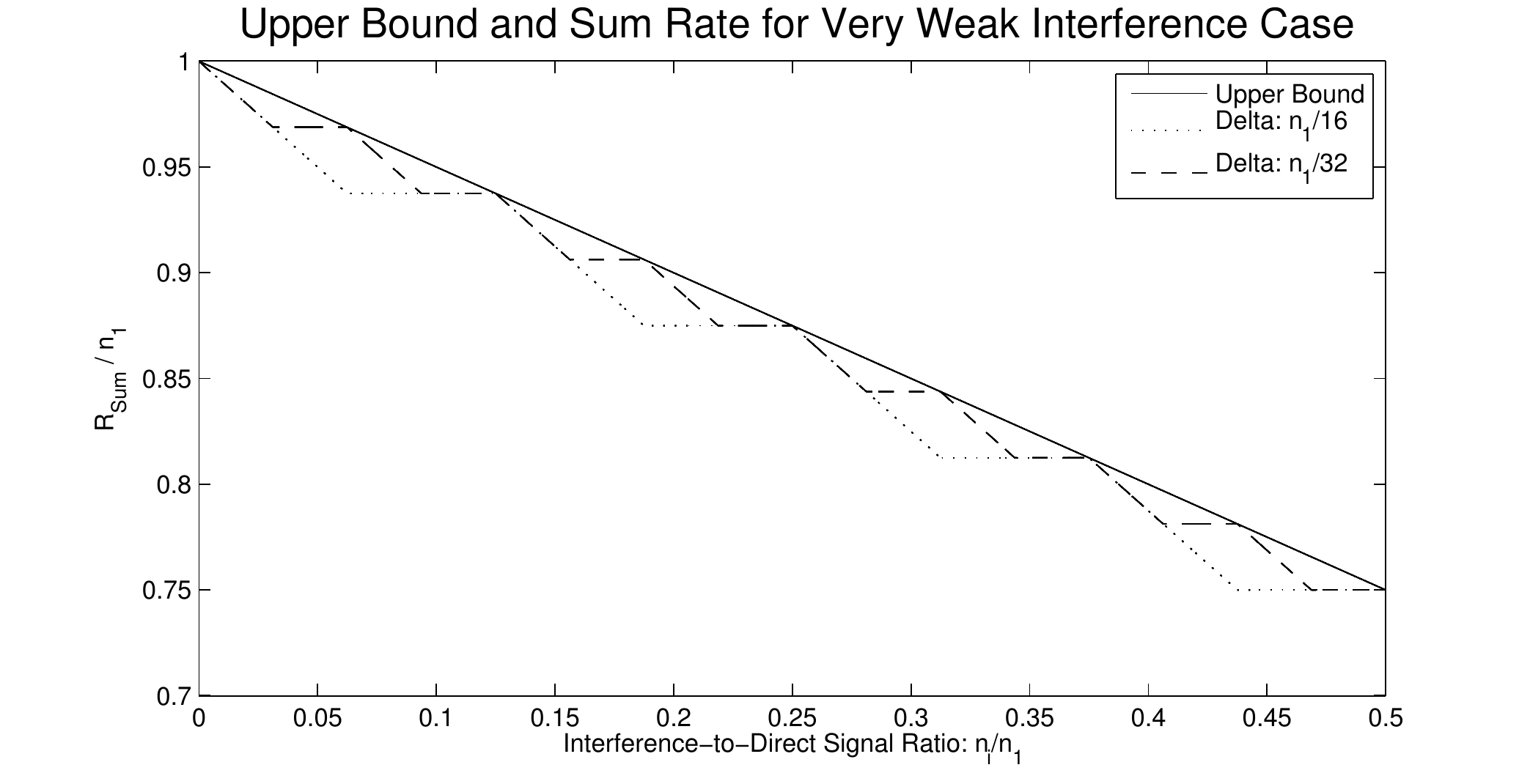}
\caption{The w-curve for the achievable schemes with $\Delta=\frac{n_1}{16}$ and $\Delta=\frac{n_1}{8}$ with the upper bound per cell. All rates and bounds are measured symmetrical, presenting a single link. Direct signals and interference is assumed to be symmetrical for clear presentation. As one can see, the achievable schemes reach the upper bound in the case when $n_i \div \Delta =$ even. The largest distance, at the peak of the triangle, is $\Delta$. As $\Delta$ becomes smaller, the achievable scheme gets closer to the upper bound.}
\end{figure}
We have shown achievable schemes and upper bounds for specific types of cellular channels in the very weak interference regime. We remark that the restriction to the very weak interference regime range is only for clarity of the exposition and future work will treat further interference ranges. For particular channel parameters the achievable schemes and upper bounds coincide and provide the sum capacity. As in the MAC-P2P case, one can see that the power-shift between the two direct signals in the cells, can be exploited to align interference at the other cell and therefore achieving higher rates than the interference channel. In other words, signal-level alignment was used to enable {\it multi-user gain}. The resulting sum capacities are the same for the IMAC and IBC and therefore show a duality relationship between them. We have also shown an expansion of the converse proof technique to the k-transmitter cellular channel, where the interference-free capacity is approached with a growing k. Further work will not only study broader interference regimes but also investigate connections to the Gaussian case, trying to accomplish a constant bit-gap result.

\bibliographystyle{./IEEEtran}
\bibliography{./ref}

\end{document}